\documentclass[journal]{IEEEtran}

\usepackage{setspace,graphicx,subfigure,caption,color,amsmath,epstopdf,cite,amssymb,url,multicol,amsfonts}
\usepackage[ruled]{algorithm2e}
\usepackage{moreverb}

\newtheorem{theorem}{Theorem}[section]
\newtheorem{lem}[theorem]{Lemma}
\newtheorem{prop}[theorem]{Proposition}

\newtheorem{conjecture}[theorem]{Conjecture}
\newtheorem{proof}[theorem]{Proof}
\newcommand{\EP}{\mbox{EAST$^+$}}
\def\argmin{\mathop{\arg\,\min}\limits}
\ifCLASSOPTIONcompsoc
\else
\fi

\ifCLASSINFOpdf
 \else
 \fi

\begin{document}

\title{Signal Reconstruction from Rechargeable Wireless Sensor Networks using Sparse Random Projections}


\markboth{IEEE Sensors Journal,~Vol.~XX, No.~XX, Month~20XX}%
{Shell \MakeLowercase{\textit{et al.}}: Bare Demo of IEEEtran.cls for Computer Society Journals}


\author{Rajib~Rana, Wen Hu,~\IEEEmembership{Senior Member,~IEEE,} Chun Tung Chou,~\IEEEmembership{Member,~IEEE} 
\thanks{Rajib Rana is with the department of Computational Informatics, CSIRO, Australia (rajib.rana@csiro.au) }%
\thanks{Chun Tung Chou is with the department of Computational Informatics, CSIRO, Australia (mohan.Karunanthi@csiro.au)}%
\thanks{Wen Hu is with the department of Computational Informatics, CSIRO, Australia (e-mail: wen.hu@csiro.au).}}

\maketitle
\begin{abstract}
Due to non-homogeneous spread of sunlight, sensing nodes possess non-uniform energy budget in rechargeable Wireless Sensor Networks (WSNs). 
An energy-aware workload distribution strategy is therefore necessary to achieve good data accuracy subject to  energy-neutral operation. Our previously proposed Energy Aware Sparse approximation Technique (EAST) can approximate a signal, by adapting sensor node sampling workload according to solar energy availability. However, the major shortcoming of EAST is that it does not guarantee an optimal sensing strategy. In other words EAST offers energy neutral operation, however it does not offer the best utilization of sensor node energy, which compromises the reconstruction accuracy.  In order to overcome this shortcoming, we propose EAST$^+$ which, maximizes the reconstruction accuracy subject to energy neutral operations. We also propose a distributed algorithm for  \EP, which offers accurate signal reconstruction with limited node to-base communications. 
\end{abstract}

\IEEEpeerreviewmaketitle

\section{Introduction}
Wireless Sensor Networks (WSNs) are currently deployed to monitor micro-climate data from different environments \cite{duckisland,springbrook}. The Springbrook National Park WSN is one such example~\cite{springbrook}. The Springbrook site is part of a World Heritage precinct in Queensland, Australia. CSIRO\footnote{The Commonwealth Scientific and Industrial Research Organization}, in partnership with the Queensland Government Environmental Protection Agency (EPA), has deployed a WSN of $200$ nodes at Springbrook in $2011$. This partnership aims to collect microclimate data for enhancing knowledge of the rain forest restoration process.

Energy supply is a major design constraint in the Springbrook deployment. In the last few years, a large number of research has been conducted (see \cite{energy_conserve} for the comprehensive list) to reduce the radio's duty cycle. However, recently it has been reported that many real life applications require specific sensors whose power consumption is significant. In addition, longer acquisition times of some specific sensors may even result in significantly higher energy consumptions than the radio transceiver.
In order to cope with the increasing energy demand, a number of sensor deployments are adopting a complementary approach of supplementing the energy supply of the system by harvesting additional energy from the environment~\cite{springbrook,tosn_wen,zhao2011joint,yerva2012grafting}.
%

Out of the variety of energy harvesting modalities, we use solar energy harvesting modality. Solar energy provides one of the highest power densities~\cite{solar}, however, sun light is not homogeneously spread over the network, which results in non-homogeneous energy level in the sensing nodes. Therefore, a sensing task allocation technique that assumes uniform energy profile of the sensing nodes could deplete the energy of a number of nodes and create holes in the network connectivity or coverage. In order to avoid such situation, an early attempt in the Springbrook deployment reduces the fraction of time the sensors are turned on to take samples (we refer this quantity as \emph{sensor on-time}) to less than $2\%$ for all nodes, which results in poor approximation of the signal.

Data collected from wireless sensor deployments are typically correlated and therefore compressible~\cite{cws} in an appropriate transform. 
If the data is compressible, a signal vector with $\hat{N}$ data values can be well approximated using only $k(<<\hat{N})$ orthonormal transform coefficients. If these $k$ largest coefficients could be determined from a small number of measurements, where measurements are collected with high probability from energy-rich sensing nodes and with smaller probability from energy-constrained nodes, we can approximate the signal with good accuracy at energy neutral condition. An energy neutral operation means that the energy consumption should be less than the energy harvested from the environment. Recently proposed techniques based on compressive sensing~(\cite{sivapalan2011compressive,cws,xu2011dynamic,chew2012sparse,rana2011adaptive,wei2012distributed}) have so far assumed that the signal is sampled uniformly. However, in order to approximate a signal with good accuracy from the rechargeable WSNs, non-uniform sampling strategies need to be developed. Our recently proposed Energy-aware Sparse Approximation Technique (EAST)~\cite{rana2010energy} offers a non-uniform sensing strategy for rechargeable wireless sensor networks, however, it is not optimal. In other words, it does not offer the best accuracy by maximizing the energy utilization of the nodes at energy-neutral condition. Therefore, methods need to be developed to maximize reconstruction accuracy subject to energy neutral operations. In this paper, we address these challenges. Our contributions are as follows:
\begin{enumerate}\itemindent 10pt
\item  We derive \EP, which offers a {\bf non-uniform} and  {\bf optimal} sensing strategy to minimize the approximation error, while preserving the energy neutral operation. 
\item We translate \EP into a distributed algorithm which can be readily implemented in the sensor field to reconstruct signal using very small number of node-to-base communications.

\item  We conduct both analytical and empirical studies to demonstrate that \EP offers the optimal sensing strategy.

\end{enumerate}

The reminder of the paper is organized as follows. In Section~\ref{sec:setup} we formulate the problem and provide necessary contextual information required for problem formulation. Then in the next section (Section~\ref{sec:eastPlus}) we describe \EP.
We describe a distributed algorithm for \EP in Section~\ref{distributed} and provide the evaluation results in Section~\ref{sec:evaluation}. Finally, we discuss the related literature in Section~\ref{sec:related_work} and then conclude in Section~\ref{sec:conclusion}.

\section{Problem Formulation}
\label{sec:setup}
Consider a WSN with $N$ nodes. Let $u\in \mathbb{R}^{M \times N}$ be the signal matrix where $u_{(h,j)}$ is the measurement of sensor node $n_j$ at time $t_h$. Assume that the network is rechargeable using solar energy. Define $E_j$ to be the amount of energy harvested by node $n_j$ during $t_{1 \leq h \leq M}$. In the rest of the paper we refer to $E_j$ as the energy profile of the node. Due to non-homogeneous spread of sunlight, $E_j$ of different nodes can be very different. For example, nodes in the open space can have higher $E_j$ whereas nodes in the forest can have smaller $E_j$. Let us define an indicator variable
\begin{equation}
\text{$f_{hj}$ = }
\begin{cases} \text{1,}  & \text{if sensor $n_j$ is turned on at time $t_h$}
\\
\text{0,}  & \text{otherwise.}
\end{cases}
\end{equation}
In order to ensure an energy neutral operation, we turn on sensor $n_j$\footnote{Note that, in WSN literature, a sensor can be used to refer to a \emph{sensor node} (which includes a CPU, a radio and measurement sensors) or a \emph{measurement sensor} (e.g. a temperature sensor, a wind speed sensor). In this paper, we refer to \emph{turning on sensor $n_j$} as to turning on the measurement sensor on node $n_j$} based on its energy profile $E_j$. Therefore, some of the values of $f_{hj}$ could be zero. The value of the signal $u$ at time instances $t_h$ where $f_{hj} = 0$ are not measured. Therefore, we need a method to estimate those components in $u$ that have not been measured.

In order to further explain the problem, in the following we will define two terms: compressible data and sparse random projections. 
For ease of presentation, we will assume $M = 1$ for the rest of this section as well as in Section~\ref{sec:eastPlus}. This means that $u$ is a $1$-dimensional vector and the $j$-th component of $u$ is the measurement from sensor $n_j$. 

\subsection{Compressible Data}
Data collected from wireless sensor deployments are typically correlated and therefore compressible in an appropriate transform~\cite{cws}, such as, the Wavelets or the Discrete Fourier Transform. Let us consider a transform $\Psi \in \mathbb{R}^{N \times N}$ whose columns form a set of $N$ orthonormal basis vectors $\{\psi_1,...,\psi_N\}$. The transform coefficient vector of the signal $u \in \mathbb{R}^N$ is given by $\Psi^T u$ where $^{T}$ denotes matrix transform. The signal $u$ is compressible, if the reordered transform coefficients $\theta_\pi$
decay like power law~\cite{practical_recovery}.  That is the $\pi$-th largest transform coefficient satisfies
\begin{eqnarray}
|\theta|_{(\pi)} \leq R \pi^{-\frac{1}{s}},  \forall: 1 \leq \pi \leq N
\label{compressibility}
\end{eqnarray}
Here $R$ is a constant, and $0 < s \leq 1$. We will call $s$ the \emph{compressibility parameter}.

Given a signal vector $u$ is compressible, the largest $k$ transform coefficients capture most of the signal information in the following sense: let the vector $\hat{\theta}_k$ be obtained from setting the smallest $(N-k)$ coefficients in $\theta$ to zero and let $\hat{u}^k = \Psi \hat{\theta}_k$, then $u \approx \hat{u}_k$ provided that the $(N-k)$ coefficients that are set to zero have small magnitude compared with those $k$ coefficients that are retained. Thus, if a signal is known to be compressible in a particular transform domain, the signal can be well approximated by recovering only the $k$ largest transform coefficients. The approximation $\hat{u}_k$ that keeps the $k$ largest transform coefficients and discards the remaining as zero is called the \emph{best $k$-term approximation}. 

The underlying hypothesis of our proposed sensing strategy is that the data collected at the energy-constrained nodes are correlated to the data collected at the energy-rich nodes. In order to ensure that we do not exhaust the available energy on the energy-constrained nodes, our framework will demand the energy-rich nodes to sample more often the energy-constrained nodes. This means that some of the measurements from the energy-constrained nodes are not available and have to be estimated. We do this by using the available measurements to estimate the $k$-largest coefficients of the complete signal in an appropriate transform domain. Once these $k$ coefficients are available, the unmeasured data can be estimated because these $k$ coefficients capture most of the signal information. 
Therefore, precisely, we want: (1) To adjust the sampling rate of a node according to its energy profile; and, (2) To be able to estimate unmeasured data from the measured one \footnote{Note: We will also refer to this process as {\sl reconstructing} the signal vector $u$ using the terminology from compressive sensing.}. We will realize these two goals by using {\sl Sparse Random Projections} as an intermediate tool. 

\subsection{Sparse Random Projections}
A projection is defined as the dot product $\phi^T u$ of a data vector $u$ and a projection vector $\phi$. When more than one projection is taken, projection vectors are packed in the rows of a matrix and a projection matrix is formed. 
An example of a projection matrix is given in \eqref{eqn:sparse_random_projections} where $\Phi_{i,j}$ is the $(i,j)$-element of the projection matrix $\Phi \in \mathbb{R}^{\ell \times N}$ with the $i$-th row of $\Phi$ containing the $i$-th projection vector.
\begin{eqnarray}
\Phi_{i,j}=\sqrt{\rho}
\begin{cases} +1 & \text{with probability $\frac{1}{2\rho}$,}
\\
0 & \text{with probability $1-\frac{1}{\rho}$,}
\\
-1 &\text{with probability $\frac{1}{2\rho}$.}
\end{cases}
\label{eqn:sparse_random_projections}
\end{eqnarray}

The projection between the projection vector $\Phi$ and signal vector $u$ is the $\ell$-dimensional vector $\Phi u$, where $\ell$ is generally smaller than $N$. We mention earlier that the projection matrix is an intermediate tool that we use to estimate the unmeasured data from the measured ones. We will see later that $\Phi u$ can be obtained from measured data and the task is to estimate $u$ from $\Phi u$. We now move on to discuss the importance of sparsity in $\Phi$.

In~\eqref{eqn:sparse_random_projections}, the sparsity parameter $\rho$ determines sparsity of the projection matrix. For example, when $\rho$=1 the matrix is dense because all elements of the matrix $\Phi$ are non-zero. When $\rho$ equal to 3, then the matrix is sparse, since, on average two-third of the elements of the projection matrix are zero. Sparsity of the projection matrix can be used to control sensor sampling patterns. It can readily be seen that the contribution of the $j$-th element of $u$ (which is be denoted by $u_j$ and represents the measurement from sensor $n_j$) to the projection $\Phi u$ is via the $j$-th column of $\Phi$. For a sparse matrix, it is possible that the $j$-th column contains all zeros, i.e. $\Phi_{ij} = 0 \forall i$, for some $j$. If this is the case, then $u_j$ is not needed to compute $\Phi u$, which in turn means that sensor $n_j$ does not need to sample.
It can be shown that the mean number of sensors that are required to sample is given by $N (1-(1-\frac{1}{\rho})^\ell)$, which can be shown to be bounded by $\frac{N \ell}{\rho}$. For $\frac{1}{\rho}=\frac{1}{N}$, this means at most $\ell$ samples are required. 

Given that there is only one parameter $\rho$ in \eqref{eqn:sparse_random_projections}, all sensors have the same probability to sample at any time instance. However, such uniform sampling strategy is not appropriate for WSNs with heterogeneous energy profile because one either has to lower the sampling rate of all nodes or leave the energy-constrained nodes in-operational for an extended period of time. In order to deal with heterogeneous energy profile, in the next section, we will generalize the sparse projection matrix so that each sensor can control its sampling probability depending on its energy profile. 

To this end we seek to achieve two key objectives: first, we want to model a projection matrix ($\Phi$)
to ensure energy neutral operation. 
Second, given the projection matrix ($\Phi$), we want to 
formulate a method for successful reconstruction of the sensed phenomena.

\section{Optimal Signal Reconstruction - \EP}
\label{sec:eastPlus}
\subsection{Modelling Projection Matrix $\Phi$ }
\label{sec:pro_matrix}
Consider a projection matrix $\Phi \in {\mathbb R}^{\ell \times N}$ whose elements $\Phi_{ij}$ has the following probability mass function:
\begin{equation}
\Phi_{ij}=\sqrt{\frac{1}{ g_j}}
\begin{cases} +1 & \text{with prob. $\frac{ g_j}{2}$}
\\
0 &\text{with prob. $1-{g_j}$}
\\
-1 &\text{with prob. $\frac{ g_j}{2}.$}
\end{cases}
\label{eqn:sparse_approx}
\end{equation}
Here $g_j=\frac{E_j}{\Sigma_{j=1}^N E_j}\kappa$ defines the probability of a measurement from sensor $n_j$
to be included in the $i$-th projection. The parameter $\kappa$ is referred to as the \emph{sampling parameter}. It is bounded by: $0<\kappa\leq1$. Probability of measurement ($g_j$) is proportional to the energy profile ($E_j$), therefore, higher energy profile of a node will increase the probability of inclusion of measurement from the node. The parameter $g_j$ is also proportional to the sampling parameter $\kappa$. In Section~\ref{sec:non-uni-sampling} we determine the optimal value of $\kappa$ that minimizes the reconstruction error. 

If $\Phi_{ij}\ne 0$, we want measurement from sensor $n_j$ to be included in the $i$-th projection. Therefore, if at least one of the $\Phi_{ij}$ is non-zero, then sensor $n_j$ will need to collect a sample. It can be shown that this happens with a probability of $1 - (1 - g_j)^\ell$. 
\subsection{Derivation of \EP}
\label{sec:non-uni-sampling}
%
%
In order to accurately recover the signal from sparse random projections, we use a simplified sketching decoder~\cite{distributed}. The prerequisite of successful recovery using sketching decoder is that
the signal needs to 
satisfy the peak-to-total energy condition. In this condition the ratio of the peak ($||u||_\infty$) to total ($||u||_2$) energy of the signal should be upper bounded by the parameter $\mu$.
The parameter $\mu$ is related to the compressibility of the signal. If signal $u$ is compressible in a transform with compressibility parameter $s$, then $\mu$ is given by
\begin{equation}
\frac{||u||_\infty}{||u||_2} \leq \mu=
\begin{cases} O(\frac{log N}{ \sqrt(N)}) & \text{if $s=1$}
\\
O(\frac{1}{ \sqrt(N)}) & \text{if $0<s<1.$}
\end{cases}
\label{eqn:cond_compress}
\end{equation}
This condition implies that the energy of the signal is not concentrated on only a few elements. In particular, if the signal is too sparse, sparse random projects may not work.

%
%
\begin{prop}
\label{prop:1}
Let $\Phi$ be the projection matrix given by Equation~(\ref{eqn:sparse_approx}). Define $x=\frac{1}{\sqrt{\ell}}\Phi u$ and $y=\frac{1}{\sqrt{\ell}}\Phi v\in \mathbb{R}^{\ell}$ as the random projection of two vectors $u$ and $v\in\mathbb{R}^\ell$. Expectation and variance of the inner product of $x$ and $y$ are respectively
\begin{eqnarray}
\mathbb{E}\left[x^{T} y \right] & = & u^T v \hspace{1cm}\text{and}\nonumber \\
Var \left(x^{T} y \right) & =  & \frac{1}{\ell}(\left(u^T v\right)^2+||u||_2^2||v||_2^2+ \Sigma_{j=1}^N \frac{ 1}{g_j}u_j^2v_j^2\nonumber\\&-&3\Sigma_{j=1}^Nu_j^2v_j^2) \nonumber.
\end{eqnarray}
\end{prop}

For proof see the Appendix. It can be observed that the variance of the estimation is inversely related to $g_j$. Thus, if $g_j$ is small, the estimation will have high variance. Note that $g_j$ is also proportional to the energy profile $E_j$. Therefore when all the nodes have good access to sunlight, good estimation can be produced.
In~\cite{distributed} it is shown that the variance of this estimation is controlled by the number of projections ($\ell$) only; and it is not shown how the variance will be changed if the nodes have non-uniform energy profile.

%

\begin{prop}
\label{prop:3}
Assume data $u \in \mathbb{R}^{N}$ satisfies the peak-to-total energy condition~(\ref{eqn:cond_compress}), and with
\begin{eqnarray}
\ell= 48\frac{(2+\mu^2\max_j\frac{1}{g_j})k^2{(1+\gamma)\log{N}}}{c^2\epsilon^2\eta^2}
\label{eq:ellTerm}
\end{eqnarray}
the sparse random matrix$\hspace{.1cm}\Phi \in \mathbb{R}^{\ell \times N}$ satisfies condition
\begin{eqnarray}
\mathbb{E}\left[\Phi_{ij}\right]=0 \mbox{,}\hspace{.05cm} \mathbb{E}\left[\Phi_{ij}^2\right]=1 \mbox{,}\hspace{.05cm} \mathbb{E}\left[\Phi_{ij}^4\right]=\frac{1}{g_j}.\label{eqn:cond_expection}
\end{eqnarray}
Denote $x=\frac{1}{\sqrt{\ell}}\Phi u$ as the sparse random projection of $u$ and $\Psi \in \mathbb{R}^{N \times N}$ as an orthonormal transform. Transform coefficients of $u$ in $\Psi$ is given by, $\theta=\Psi^{-1} u$. Assume the best $k$-term approximation gives an approximation ($\hat{u}_{opt}$) with error $||u-\hat{u}_{opt}||_2^2 \leq \eta||u||_2^2$. Using only $x$, $\Phi$ and $\Psi$, $u$ can be recovered with error
\begin{eqnarray}
\frac{||u-\hat{u}||_2^2}{||u||_2^2} \leq (1+\epsilon)\eta
\label{eqn:error_metric}
\end{eqnarray}
with probability at least $1-N^{-\gamma}$.
\end{prop}
For proof see the Appendix. 


Using Eq.\eqref{eq:ellTerm}, the error term can be written as, $\epsilon={(48\frac{(2+\mu^2\max_j\frac{1}{g_j})k^2{(1+\gamma)\log{N}}}{c^2\eta^2\ell}})^{1/2}$.
Rrecall that $g_j=\frac{E_j\kappa}{\Sigma_{i=1}^N E_i}$, therefore $max_j \frac{1}{g_j}=\frac{\Sigma_{i=1}^N E_i}{E_{\min}}$, where $E_{\min}$ represents the energy profile of the node with minimum energy. Let $\frac{48(1+\gamma)k^2}{c^2\eta^2}=c_1$ and $\mu^2=c_2$, therefore, the approximation error can be rewritten as: 
\begin{eqnarray}
\epsilon&=&({\frac{c_1\log{N}}{\ell}(2+c_2\frac{\Sigma_{i=1}^N E_i}{E_{\min}\kappa})})^{1/2}.\label{old_opt}
\end{eqnarray}
Our objective is to minimize the approximation error while ensuring energy neutral operation. It can be written as follows:
\begin{eqnarray}
\argmin_{\ell,\kappa}
\frac{c_1\log{N}}{\ell}(2+c_2\frac{\sum_{i=1}^N E_i}{E_{\min}\kappa})\label{objective_function}
\end{eqnarray}
Subject to:
\begin{eqnarray}
(1-(1-\frac{E_{1}\kappa}{\sum_{i=1}^N E_i})^{\ell})c_4&\leq& E_{1} \label{con_1}\\
.\nonumber\\
.\nonumber\\
(1-(1-\frac{E_{N}\kappa}{\sum_{i=1}^N E_i})^{\ell})c_4&\leq& E_{N}.\label{con_3}
\end{eqnarray}
The objective function (\ref{objective_function}) aims to minimize the approximation error. The constraints (\ref{con_1})--(\ref{con_3}) keep each of the nodes within their energy budget. 
The term $(1-(1-\frac{E_1\kappa}{\sum_{i=1}^N E_i})^{\ell})$ is the probability that the node with the energy profile $E_1$ acquires a sample. The constant $c_4=VIT$ is the energy required to acquire a sample where, $V$ is the battery voltage and $I$ and $T$ are the electrical current and  time to acquire a sample, respectively. In order to ensure energy neutral operation, we want to ensure that average energy consumed for sampling $(1-(1-\frac{E_{j}\kappa}{\sum_{i=1}^N E_i})^{\ell})c_4$ is less than the harvested energy $E_j$.

It is possible to reduce the number of constraints in the above optimization problem if we could find a constraint, such that, if this constraint is satisfied, it implies that all other constraints are also satisfied. In other words, we have to find an active constraint. Using the general representation ($E_j-c_4(1-(1-\frac{E_j\kappa}{\sum_{i=1}^N E_i})^{\ell}) \geq 0$) of the constraints (\ref{con_1})-(\ref{con_3}), the number of measurements $\ell$ can be written as 
\begin{eqnarray}
\ell\leq\frac{\log{(1-\frac{E_j}{c_4})}}{\log{(1-\frac{E_j\kappa}{\sum_{i=1}^NE_i})}}.\label{con_l}
\end{eqnarray}
Without loss of generality, we assume that the nodes are ordered in non-decreasing order of energy profiles, i.e. $E_1\leq E_2...\leq E_N$. Under this assumption, if the right-hand-side of (\ref{con_l}) is a non-decreasing function of $E_j$, then only the constraint with $E_j = E_1$ is active. This is because, for the rest of the energy profiles $E_2$ to $E_{N}$, $\ell$ will be less than or equal to the right-hand-side of (\ref{con_l}) for $E_j = E_1$. Similarly, if the right-hand-side of (\ref{con_l}) is a non-increasing function of $E_j$, then the constraint with $E_j=E_N$ would be sufficient. 
In order to determine whether the right-hand-side of (\ref{con_l}) is non-increasing or non-decreasing, we find the derivative of (\ref{con_l}) with respect to $E_j$. If the derivative is positive, the right-hand-side of (\ref{con_l}) is non-decreasing and if it is negative, the right-hand-side is non-increasing. The derivative of the right-hand-side of (\ref{con_l}) is given by 

\begin{eqnarray}
&-&\underbrace{\frac{\log{(1-\frac{E_j}{C_4})(\frac{E_j\kappa}{(\sum_{i=1}^N E_i)^2}-\frac{\kappa}{\sum_{i=1}^N E_i})}}{(1-\frac{E_j\kappa}{\sum_{i=1}^N E_i})\log^2{(1-\frac{E_j\kappa}{\sum_{i=1}^N E_i})}}}_{\mbox{left part}}\nonumber\\&-&\underbrace{\frac{1}{c_4(1-\frac{E_j}{c_4})\log{(1-\frac{E_j\kappa}{\sum_{i=1}^N E_i})}}}_{\mbox{right part}}.\label{deri}
\end{eqnarray}

\begin{conjecture}
\label{cojec:5}
\emph{Let us assume that the harvested energy is less than the consumed energy\footnote{Note that if harvested energy 
is higher than consumed energy, i.e. $E_j > c_4$, then , the problem
is trivial, since energy neutral operation is 
automatically satisfied. 
We instead consider the non-trivial case where harvested energy is less than or equal to 
consumed energy. } i.e. $\forall_j:\mbox{ }E_j<c_4$.} Then the derivative in \eqref{deri} is greater than $0$.
\end{conjecture}


{\bf Validity of Conjecture~\ref{cojec:5}.}
The denominator of the right part of the derivative is negative: $c_4(1-\frac{E_j}{c_4})$ is positive but since $0<(1-\frac{E_j\kappa}{\sum_{i=1}^N E_i})<1$, $\log{(1-\frac{E_j\kappa}{\sum_{i=1}^N E_i})}$ is negative. Thus, the right part of the derivative is negative.

Based on the sign of the left and right part of the derivative, it can be rewritten as
\begin{eqnarray}
&&\underbrace{\frac{1}{c_4(1-\frac{E_j}{c_4})|\log{(1-\frac{E_j\kappa}{\sum_{i=1}^N E_i})}|}}_{\mbox{left part}}\nonumber\\&-&\underbrace{\frac{|\log{(1-\frac{E_j}{C_4})}||(\frac{E_j\kappa}{(\sum_{i=1}^N E_i)^2}-\frac{\kappa}{\sum_{i=1}^N E_i})|}{(1-\frac{E_j\kappa}{\sum_{i=1}^N E_i})\log^2{(1-\frac{E_j\kappa}{\sum_{i=1}^N E_i})}}}_{\mbox{right part}}
.\label{new_der}
\end{eqnarray}
Note that the sign of (\ref{new_der}) is dependent on the values of different variables and their interrelationships. Such as, the term $|\log{(1-\frac{E_j\kappa}{\sum_{i=1}^N E_i})}|$ can be either greater or less than the term $\log^2{(1-\frac{E_j\kappa}{\sum_{i=1}^N E_i})}$ based on the value of $(1-\frac{E_j\kappa}{\sum_{i=1}^N E_i})$. 
It is therefore a non-trivial exercise to determine the sign of the derivative using analytical methods. 
We instead determine the sign of the derivative using simulations.  
\begin{figure}[]
\centering
\subfigure[Right Skewed]{
\includegraphics[width=.45\linewidth]{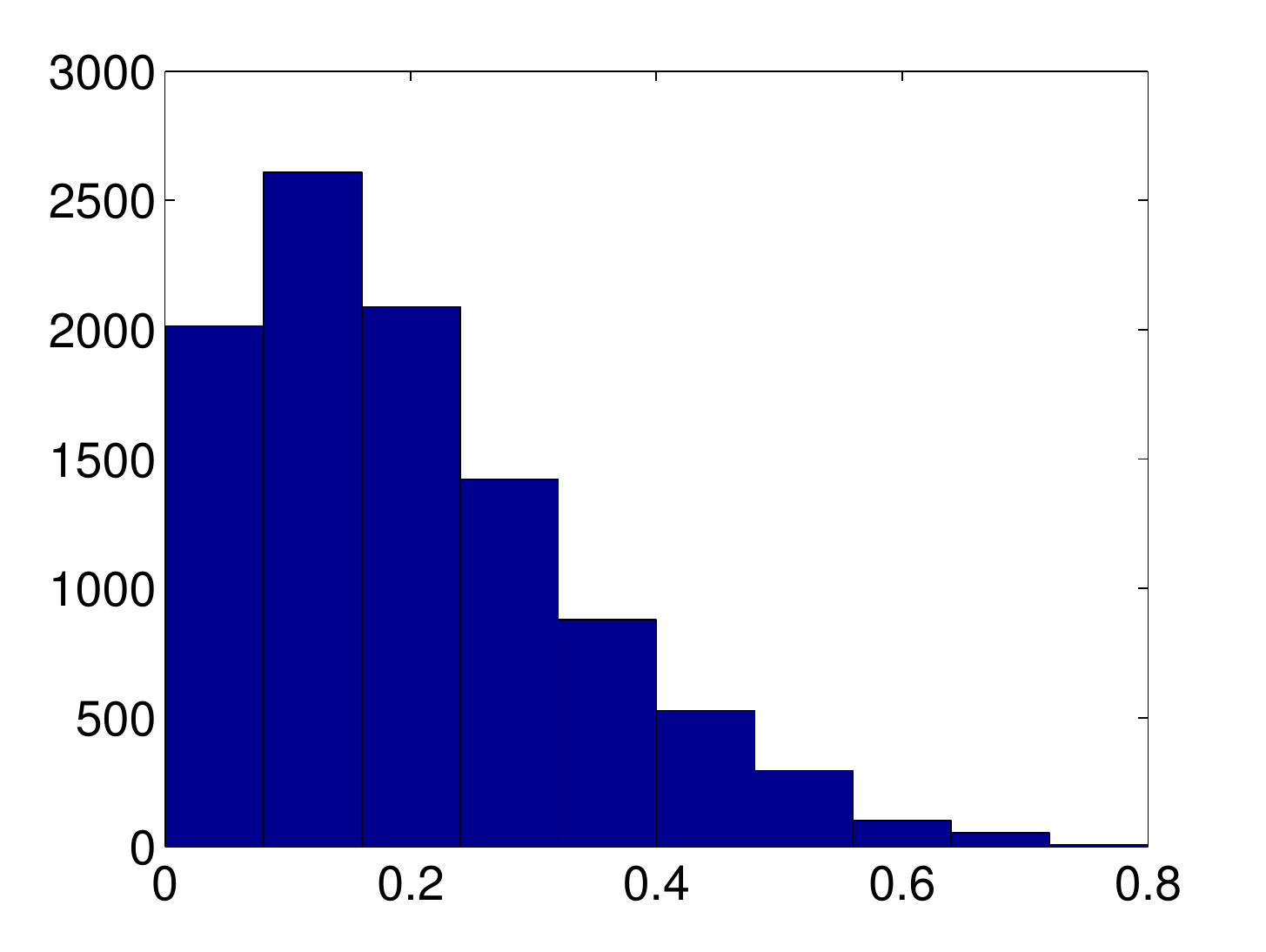}
\label{fig:right_beta}
}
\subfigure[Left Skewed]{
\includegraphics[width=.45\linewidth]{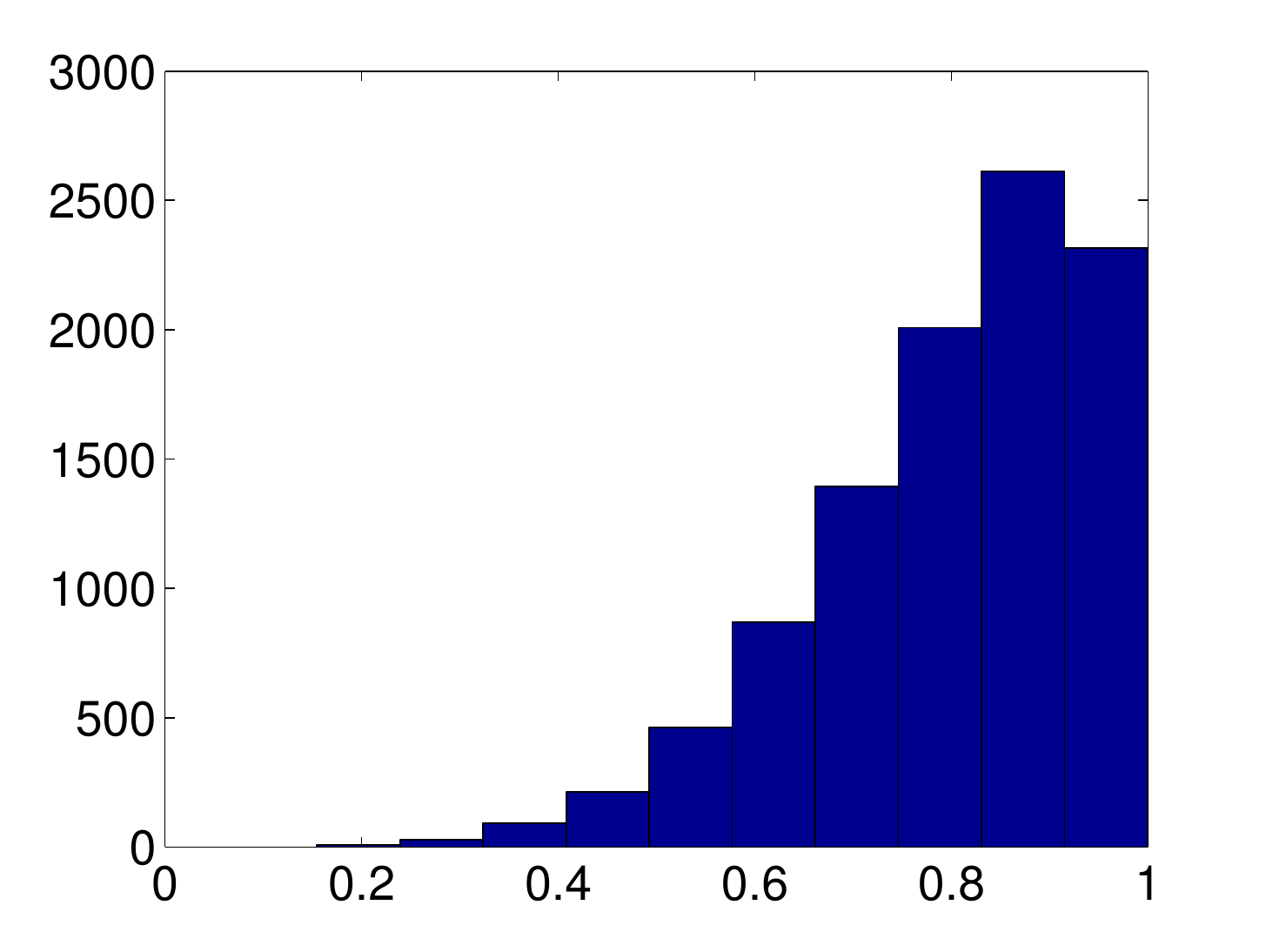}
\label{fig:left_beta}
}
\caption{Histogram of beta distribution. Along x-axis is the ratio $\frac{E_j}{\sum_{i=1}^N E_i}$ and along y-axis is the frequency.}
\label{fig:beta_plots}
\end{figure}

We model energyf profile of the nodes using beta distribution with two positive shape parameters $\alpha$ and $\beta$. Controlling the values of $\alpha$ and $\beta$ we can control the shape of the probability density function (pdf) of beta distribution. Such as, when $\beta >> \alpha$ the pdf is strongly right skewed and it is strongly left skewed if the values of $\alpha$ and $\beta$ were switched. We show the histograms of beta distribution for right and left skewed orientations in Figure~\ref{fig:right_beta} and \ref{fig:left_beta}, respectively. Using left skewed orientation of the beta distribution, a network can be modeled to contain energy-rich nodes with high probability and energy constrained nodes with small probability. On the other hand, using right skewed orientation of the beta distribution a network can be modeled to contain energy constrained nodes with high probability and energy-rich nodes with small probability.
\begin{figure}[h]
\centering
{
\includegraphics[width=.7\linewidth]{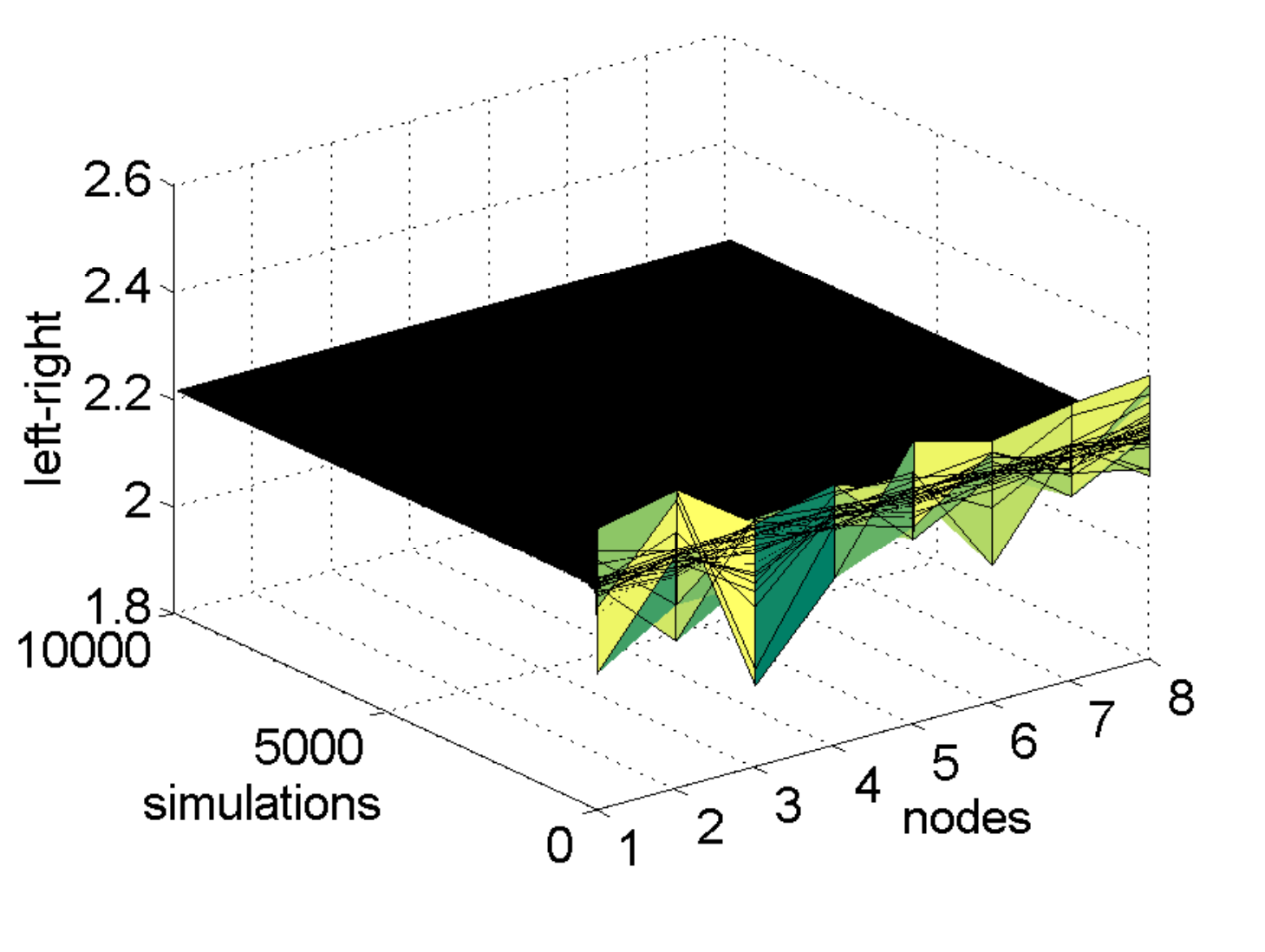}
}
\caption{Sign of the derivative.}
\label{fig:derivative_sign}
\end{figure}
Modeling $\frac{E_j (1\leq j\leq N)}{\sum_{i=1}^N E_i}$, separately using left and right skewed beta distribution, and drawing $\kappa$ uniform randomly from $(0,1)$, over $10^3$ simulations we observe that (see Figure~\ref{fig:derivative_sign}) the left part of (\ref{new_der}) is always greater than the right part (left - right along $Y$-axis is always positive). Therefore, we conjecture that the sign of the derivative is positive.

Since the sign of the derivative is positive, the right-hand-side of (\ref{con_l}) is an non-decreasing function of $E_j$. Therefore, constraint with $E_j=E_1$ is the only active constraint. We rewrite (\ref{con_l})
\begin{eqnarray}
\ell\leq\frac{\log{(1-\frac{E_{1}}{c_4 })}}{\log{(1-\frac{E_{1}\kappa}{\sum_{i=1}^NE_i})}}.\nonumber
\end{eqnarray}

Recall that the number of non-zero elements per row of the projection matrix is proportional to the sampling parameter $\kappa$. For a given energy budget, if we increase $\ell$, we have to decrease $\kappa$, otherwise due to higher sampling probability, some nodes may deplete their energy. However, if $\kappa$ is too small, there would be lot of rows in the projection matrix which will be zero. A row with all zeros contains no information of the signal and thus is useless for the signal reconstruction process. We therefore impose an additional constrain that the expected number of non-zero elements in each row to be at least one by using the constraint $g_j \geq \frac{1}{N}$.
Recall that $g_j=\frac{\kappa E_j}{\sum_{i=1}^N E_i}$, therefore,
\begin{eqnarray}
\kappa\geq \mbox{ }\frac{\sum_{i=1}^N E_i}{N E_j},\mbox{ }\forall_j. \label{kappa_expession}
\end{eqnarray}
Note that (\ref{kappa_expession}) is a collection of $N$ constraints, but it can readily be shown that these $N$ constraints can be replaced by the following single constraint:
\begin{eqnarray} 
\kappa\geq\frac{\sum_{i=1}^N E_i}{N E_{1}}. \label{new_kappa}
\end{eqnarray}

With Conjecture \ref{cojec:5} and the new constraint (\ref{new_kappa}), we re-write the optimization problem as:
\begin{eqnarray}
\argmin_{\ell,\kappa}
\frac{c_1\log{N}}{\ell}(2+c_2\frac{\sum_{i=1}^N E_i}{E_{\min}\kappa})\label{opt_prob}
\end{eqnarray}
Subject to:
\begin{eqnarray}
\ell-\frac{\log{(1-\frac{E_{1}}{c_4 })}}{\log{(1-\frac{E_{1}\kappa}{\sum_{i=1}^NE_i})}}\leq0\nonumber\\
\kappa-\frac{\sum_{i=1}^N E_i}{N E_{1}}\geq0.\nonumber
\end{eqnarray}
Solution to \eqref{opt_prob} provides the {\bf optimal} value of $\ell$ and $\kappa$.

\section{Distributed Algorithm}
\label{distributed}
We design a distributed algorithm for \EP where nodes {\bf locally} generate projections without communicating with base and thus save the additional energy required by the centralized approach for base to node communication. Our description so far has assumed $M = 1$, however \EP can be readily extended to the case with $M > 1$. In this case, we consider the sensor measurement $u_{hj}$ collected at time $t_h$ ($h = 1,...,M$)
by sensor $n_j$ ($1\leq j \leq N$). We will also \emph{vectorize} the 2-dimensional signal $u_{hj}$. We will abuse the notation and use $u$ to denote this vector (this should be clear from the context). The vector $u$ has $\hat{N} = M N$ elements where the $q$-th element of $u$ is $u_{hj}$ where $q = h + (j-1)M$. The corresponding projection matrix $\Phi$ is now an $\ell \times \hat{N}$ matrix. For $q = h + (j-1)M$, the elements in the $q$-th column of the projection matrix ($\Phi_{iq}$ with $i = 1,...,\ell$) are generated by Equation~(\ref{eqn:sparse_approx}) with parameter $g_j$ and these elements will determine whether the sensor $n_j$ will sample at time $t_h$. We will now describe an algorithm which is used by \EP to recover an approximation of the signal ($u$), from the sparse projections created locally in different nodes.

\begin{itemize}
\item First each node $n_{\tilde{j}}$ ($1 \leq \tilde{j} \leq N$)   generates the random numbers $\Phi_{r1}, ..., \Phi_{r\hat{N}}$ using the distribution function mentioned in Equation (\ref{eqn:sparse_approx}).  Each of $n_{\tilde{j}}$ is responsible for generating the $r$-th row ($1 \leq r \leq \ell$) of the projection matrix. Consider the element $\Phi_{r q}$ in the projection matrix and let us assume that the column index $q$ and the node-time pair $(j,h)$ have one-to-one correspondence given by $q = (j - 1)M + h$. 
\item If $\Phi_{r q}\ne 0$, node $n_{\tilde{j}}$ asigns node $n_j$ to sample at time $t_h$ and node $n_j$ sends the sample to node $n_{\tilde{j}}$.
\item Upon receiving $u_{jh}$ from node $n_j$, $n_{\tilde{j}}$ computes $u_r=\Sigma_{q=1}^{\hat{N}}\Phi_{rq} u_{q}$ (where $u_q = u_{jh}$). Node $n_{\tilde{j}}$ performs this operation for all the values it receives and finally transmits $u_r$ to the base station. This process is repeated for all node $n_{\tilde{j}},1\leq \tilde{j} \leq N.$
\item After receiving transmissions from the nodes, base station has $\Phi_{\ell \times \hat{N}}u=[x_1,...,x_\ell]^T$. It then generates $\Phi_{\ell \times \hat{N}}$ using the same seed as the nodes. Finally, with $x(=\Phi_{\ell \times \hat{N}}u)$, $\Phi_{\ell \times \hat{N}}$ and $\Psi$, base station uses low-complexity sketching decoder to recover the signal. The complexity of the decoder is $O(\ell\hat{N}\log{\hat{N}})$.
\end{itemize}
\begin{figure}[]
\centering
\includegraphics[width=\linewidth]{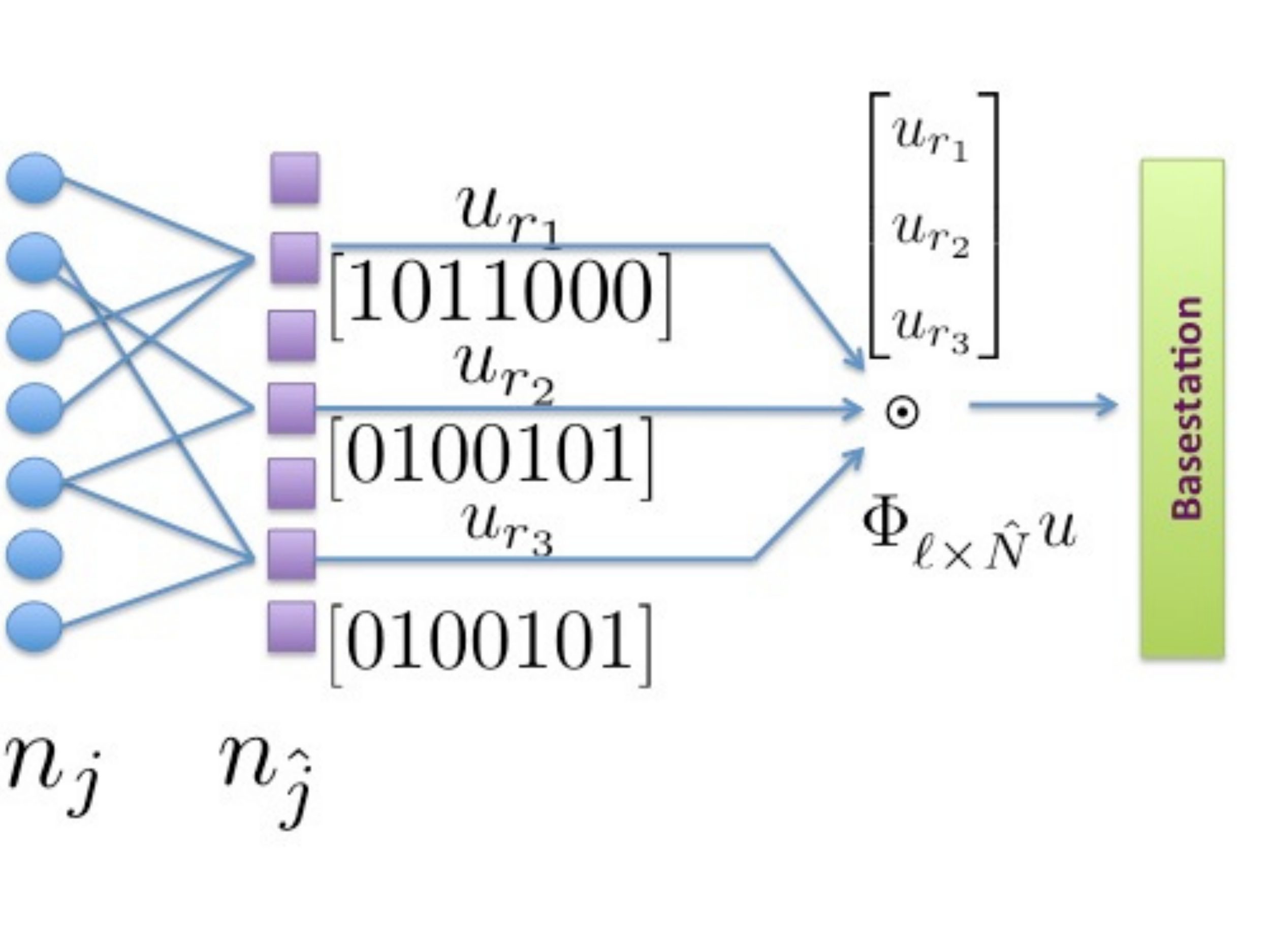}
\caption{Illustration of the Distributed Algorithm for one time snapshot.}
\label{fig:evaluation}
\end{figure}

\section{Experimental Results}
\label{sec:evaluation}
In this Section we  use the data collected from energy hungry wind speed and wind direction sensors at the Springbrook deployment to evaluate the performance of \EP. \subsection{Experimental Setup}
Let $\hat{u}$ be the approximation of the signal $u$, we use relative error, $||u-\hat{u}||_2^2/||u||_2^2$ to determine the accuracy of the approximation. The relative error is a commonly used error metric in the signal processing literature~\cite{bcs,distributed} that tells us how close the approximate signal is to the real signal.

We used data from $8$ of the sensing nodes from Springbrook National Park deployment.  Amongst these $8$ nodes (shown in Figure~\ref{fig:evaluation1}), node $5$ is deep in the forest whereas the rest of the nodes are in the open space. Consequently, the solar current harvest rate of node $5$ is the lowest whereas the rest of the nodes have higher (also similar) harvest rates. Inter-sampling interval in the deployment is $5$ minute. We collected one month of data which gave us $8640$ snapshots of both wind speed and wind direction sensor data. Note that sketching decoder computes the estimation from median, therefore, it performs better approximation with large $\hat{N}$. We used $\hat{N}=M \times N=2048$, by segmenting our snapshots from $N=8$ nodes into group of $M=256$ each. 
Below we define a number of variants of \EP that we will use in the results.
 \begin{figure}[]
\centering
\includegraphics[width=0.44\linewidth]{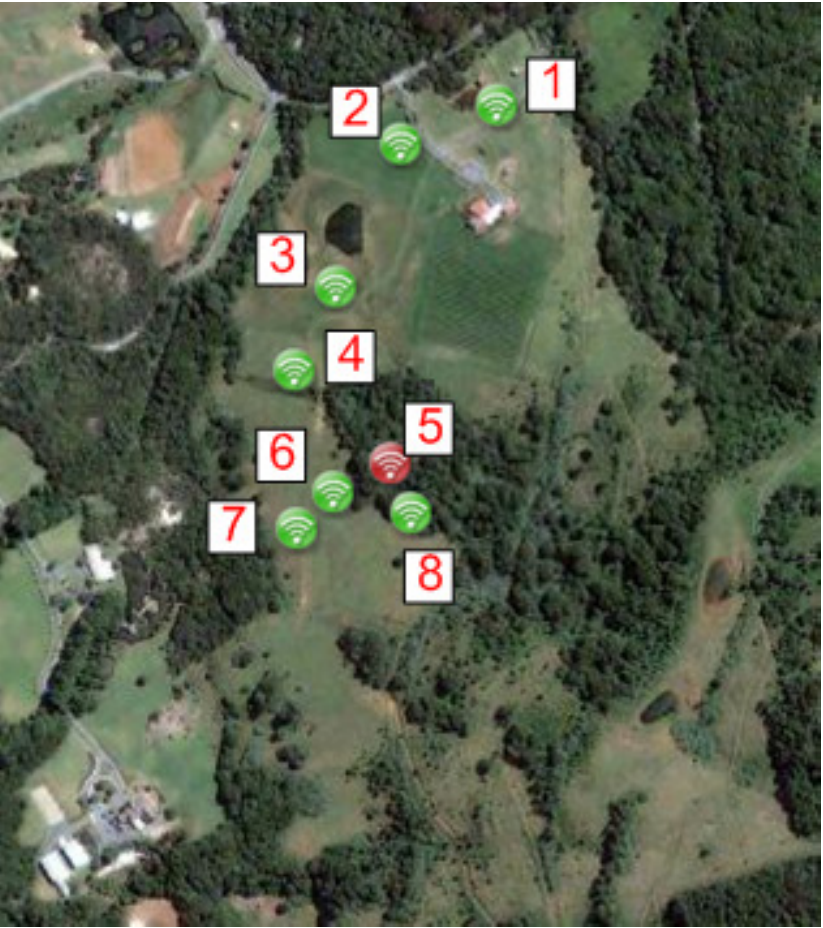}
\caption{ Location of the Springbrook sensing nodes.}
\label{fig:evaluation1}
\end{figure}

{\bf \EP-Upperbound:} Given the energy profile of the nodes, we solve the optimization problem in~(\ref{opt_prob}) to determine this minimum value of the upper bound of the approximation error.  We call this error \EP-Upperbound. We use the Matlab optimization solver ``fmincon'' to solve~(\ref{opt_prob}). 

{\bf \EP-Empirical:} The difference between this and \EP-Upperbound is that, unlike \EP-Upperbound, the error is not given by the solution of~\eqref{opt_prob}. The error is calculated from the reconstruction of the real data, such as wind speed and wind direction data. However, for various data length  $\hat{N}$, we use the optimal $\ell$ and $\kappa$ given by \EP-Upperbound.

{\bf EAST-Equality:} \EP uses inequality constraints 
$\forall_j: \mbox{ }(1-(1-g_j))^{\frac{1}{\ell}}c_4\leq E_j$ 
and we conjecture that the optimal solution given by 
\EP requires only one constraint to be active.
Precisely, we conjecture that the optimal solution requires only the node with minimum energy to operate at ``exact'' energy neutral operation.
Let us now envision another solution 
where all the constraints are active, i.e.
\begin{eqnarray}
\forall_j:\mbox{ }(1-(1-g_j)^{\ell})c_4= E_{j}.\label{east_equality}
\end{eqnarray}
We will refer to this variant as EAST-Equality. 

\subsection{Results}
\begin{figure}[]
\centering
\subfigure[Wind Speed]{
\includegraphics[width = 0.45\linewidth]{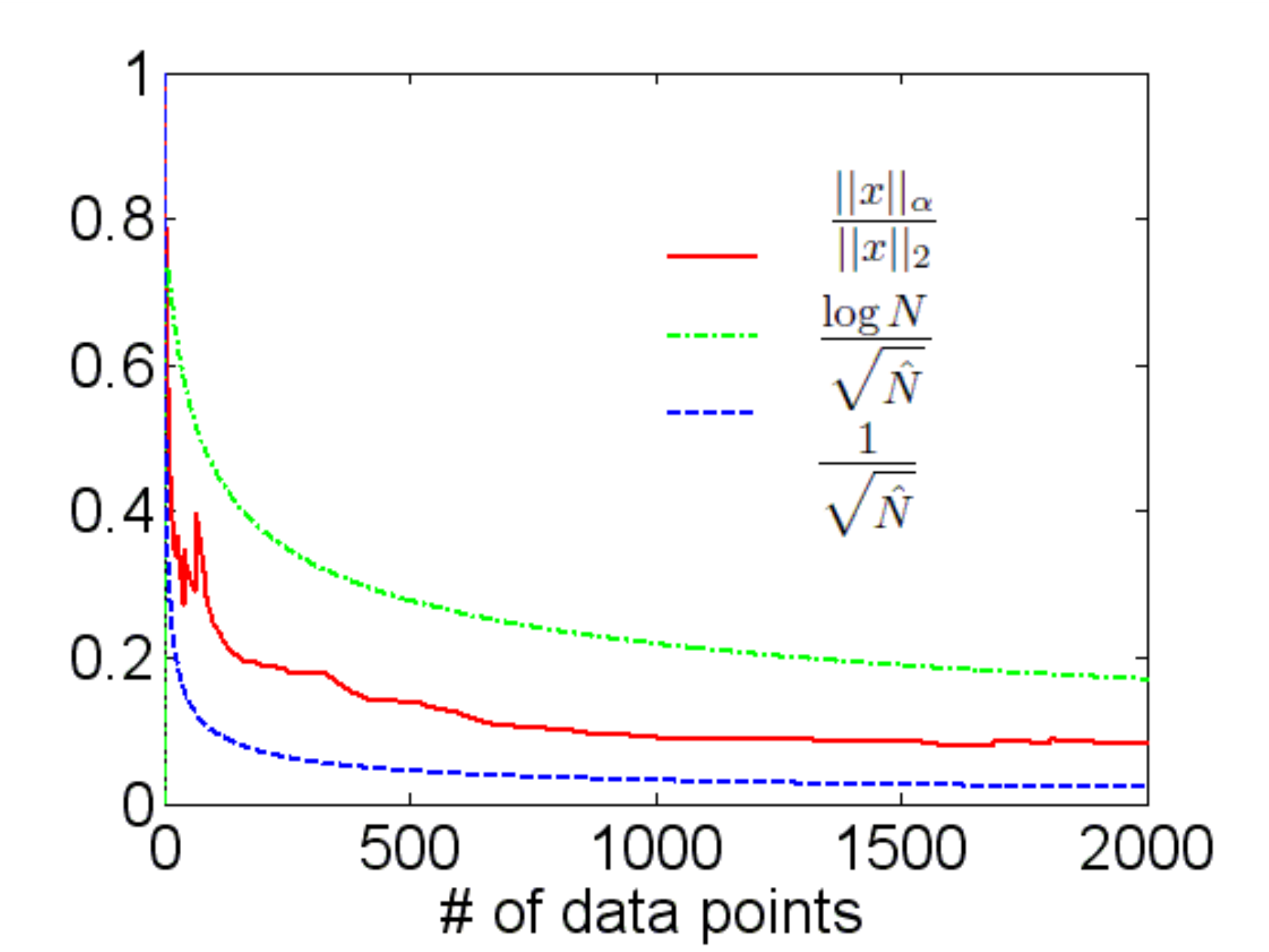}
\label{fig:wind_speed_basis}
}
\subfigure[Wind Direction]{
\includegraphics[width = 0.45\linewidth]{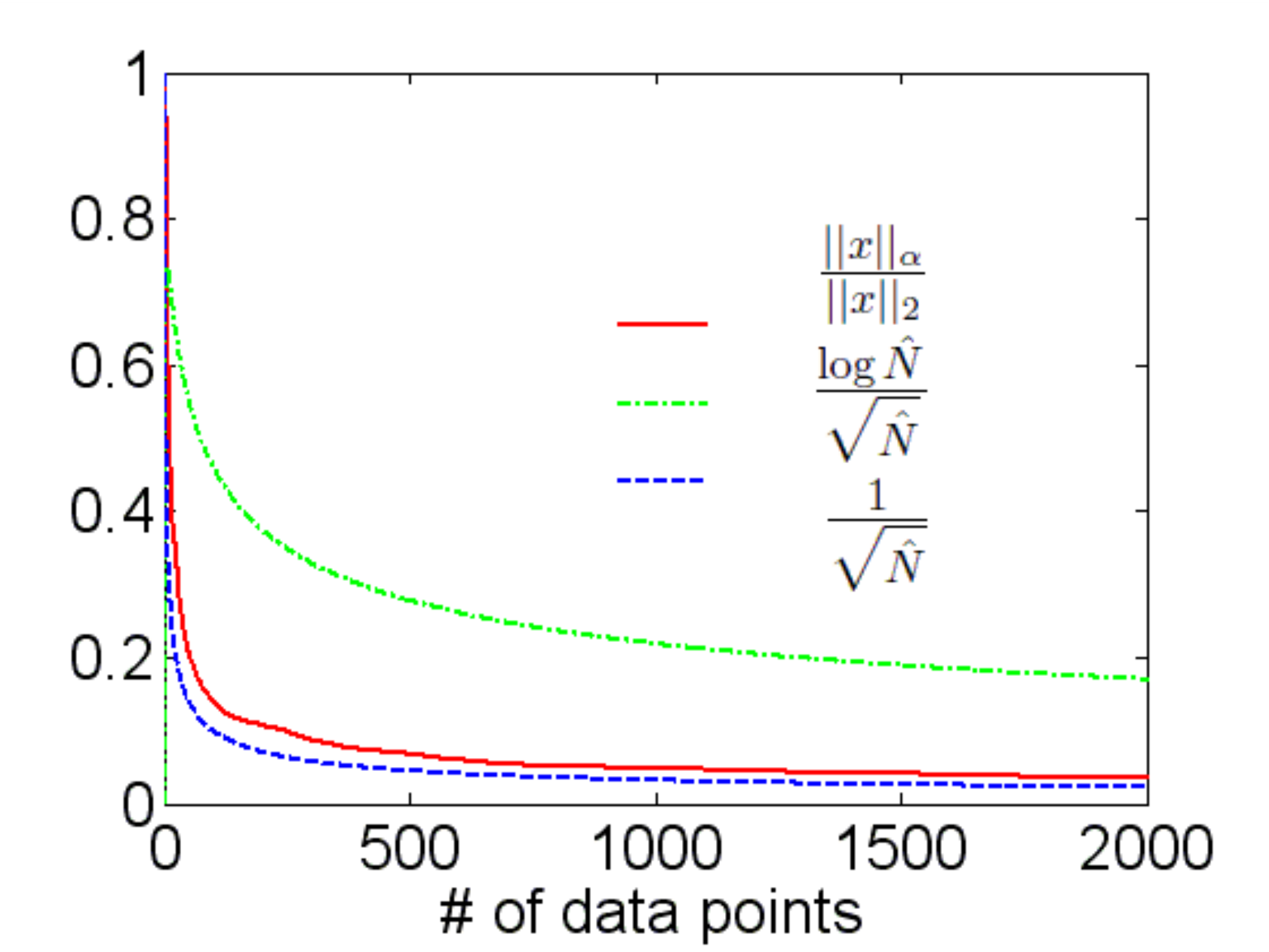}
\label{fig:wind_direction_basis}
}
\caption{Peak-to-total energy condition on data.}
\label{fig:basis}
\end{figure}
In the results section we first verify whether the wind data satisfy the bounded peak-to-total energy condition. In Figure~\ref{fig:basis} we plot the peak to total energy ratio  $\frac{||u||_{\infty}}{||u||_2}$ for various $\hat{N}$. We observe that for both of the wind sensor data, $\frac{||u||_{\infty}}{||u||_2}$ is bounded by $\frac{\log{\hat{N}}}{\sqrt{\hat{N}}}$ and $\frac{1}{\sqrt{\hat{N}}}$. This satisfies the peak-to-total energy condition. 


\begin{figure*}[tp]
\centering
\subfigure[Wind Speed]{
\includegraphics[width=.4\linewidth]{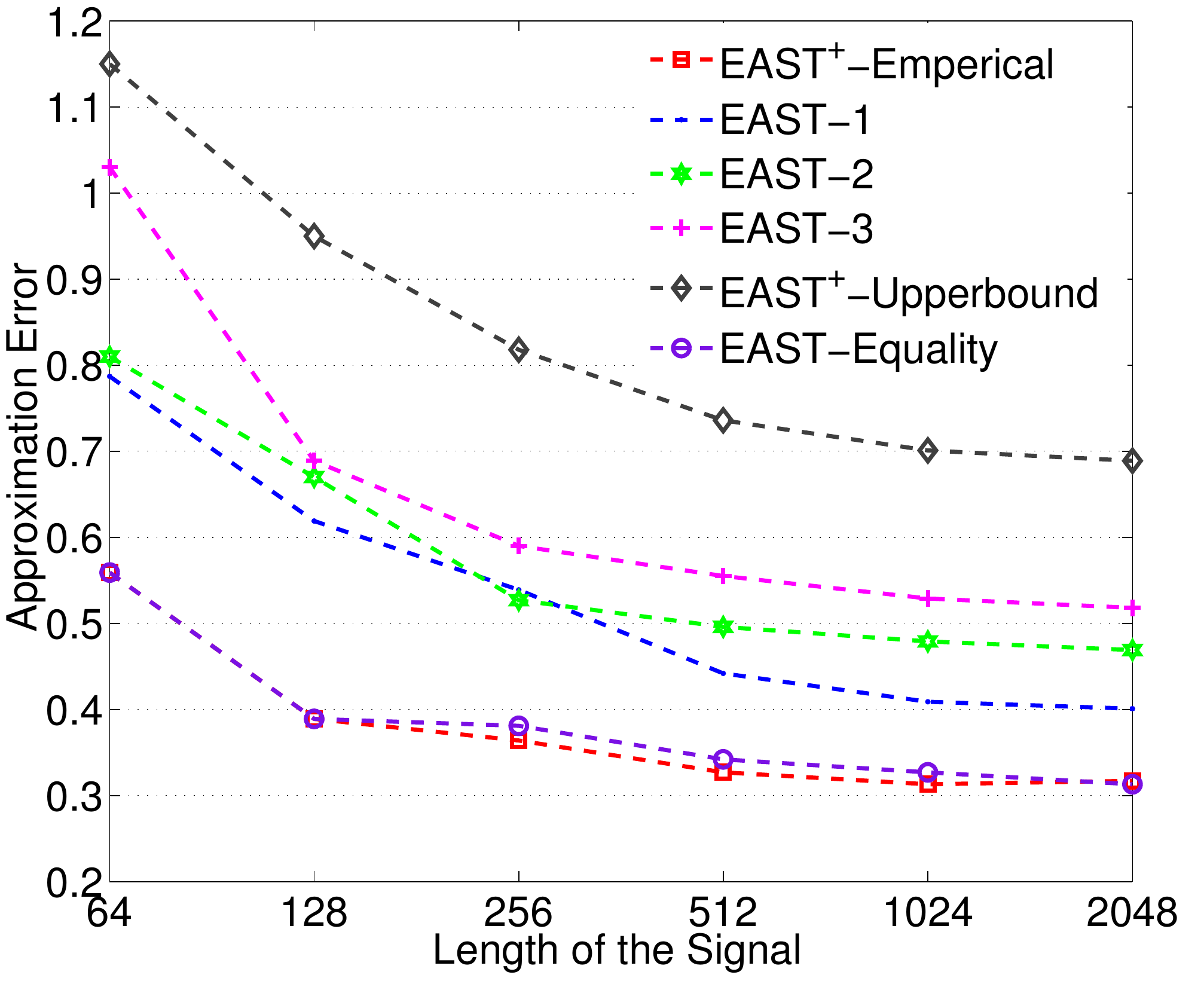}
\label{fig:wind_speed_optimization}
}
\subfigure[Wind Direction]{
\includegraphics[width=.4\linewidth]{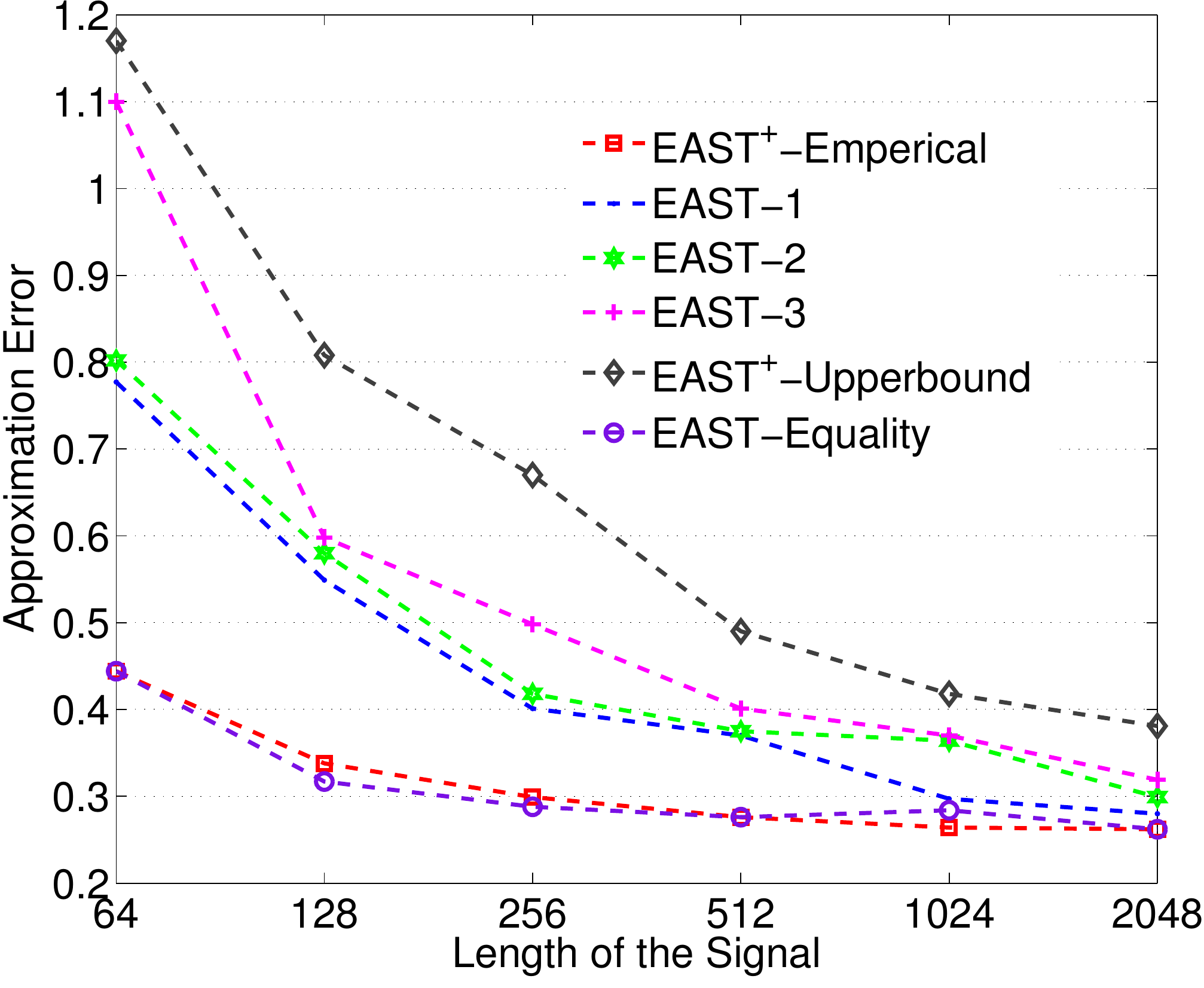}
\label{fig:wind_dir_optimization}
}
\caption{Evaluation of \EP.}
\label{fig:theory_vs_empirical}
\end{figure*}
%
%
We report the performance of \EP in Figure~\ref{fig:theory_vs_empirical}. In general, the approximation error decreases as the length of the signal ($\hat{N}$) increases. This can be explained using the behavior of sketching decoder, which performs better when the length of the signal is large. Clearly, \EP-Upperbound provides the upper bound of the reconstruction error for all the values of $\hat{N}$. It provides the upper bound for EAST, \EP-Empirical and EAST-Equality. 

We also compare the performance of \EP with our previously proposed EAST. In particular, we compare \EP-Empirical with EAST. We choose sufficiently large values of $\ell$ and $\kappa$ for EAST, which  ensure energy neutral operations. 
In Table~\ref{tab:ell_kappa_wind_dir} we report three sets of $\ell$ and $\kappa$ used for wind direction. Similar values of $\ell$ and $\kappa$ were chosen for wind speed. The three reconstructions using these three set of values are referred to as EAST-1, EAST-2 and EAST-3 in Figure~\ref{fig:theory_vs_empirical}. We find that \EP-Empirical achieves significant error reduction compared to EAST-1, EAST-2 and EAST-3. The error reduction is highest when signal length is small. For example, for wind speed, when signal length is 64, the error reduction is approximately 17\% compared to the EAST-1 (We compare with EAST-1, since it performs better than EAST-2 and EAST-3). Similarly, for wind direction, the error reduction is approximately 30\%. However, the error reduction diminishes as the length of the signal increases. For wind speed the error reduction is 8\% when signal length is 2048. Note that for in-situ signal processing on resource improvised sensor nodes, small signal length is desirable. Therefore, \EP is highly preferable over EAST for wireless sensor networks.  

Finally, we compare \EP-Emperical with EAST-Equality. The comparison results show that these two have similar reconstruction performances. In East-Equality all the nodes operate at exact energy neutral condition. Therefore, East-Equality would require higher sampling energy compared to \EP. However, despite sampling at a higher rate,  East-Equality  does not offer any significant reconstruction improvement compared to \EP-Empirical. 
This further substantiates our claim that \EP is optimal.

\begin{table*}[t]
\centering
\caption{Choice of $\ell$ and $\kappa$ used in Figure~\ref{fig:theory_vs_empirical} for wind direction data. E1, E2, E3 and $\mbox{E}^+$ refers to EAST-1, EAST-2, EAST-3 and \EP respectively.}
\resizebox{!}{0.8cm}{
\setlength{\tabcolsep}{0.8pt}
\begin{tabular}{|c|r|r|r|r|r|r|r|r|r|r|r|r|c|r|r|r|r|r|r|r|r|r|r|r|r|} \hline
&\multicolumn{4}{c|}{$\hat{N}$=64}&\multicolumn{4}{c|}{$\hat{N}$=128}&\multicolumn{4}{c|}{$\hat{N}$=256}&\multicolumn{4}{c|}{$\hat{N}$=512}&\multicolumn{4}{c|}{$\hat{N}$=1024}&\multicolumn{4}{c|}{$\hat{N}$=2048}\\\hline
&E1  & E2  & E3   &  $\mbox{E}^+$&E1  & E2  & E3   &  $\mbox{E}^+$&E1  & E2  & E3   &  $\mbox{E}^+$&E1  & E2  & E3   &  $\mbox{E}^+$&E1  & E2  & E3   &  $\mbox{E}^+$&E1  & E2  & E3   &  $\mbox{E}^+$ \\\hline
$\ell$&413  &290   &370   &368  &850 & 660 &720   &740 &1500 & 1290 & 1460  &1482&3295  & 2675  & 2950   &2968  &6370  & 5150  & 5910   & 5940 &12200 & 11190 & 11780   &  11883  \\\hline
$\kappa$&0.06  &0.1   &0.08   &0.0945  &0.01 &0.05  &0.02   &0.0473 &0.009 &0.03  & 0.0239  & 0.0236 &0.009  &0.02   &0.008    &0.0118  &0.0023  &0.0071   & 0.0047   & 0.0059 &0.0013 &0.0047  &0.0021    & 0.0030  \\\hline
\end{tabular}
}
\label{tab:ell_kappa_wind_dir}
\end{table*}

\section{Related Work}
\label{sec:related_work}

Energy conservation in wireless sensor networks is a heavily studied field. Below we rationalize the uniqueness of \EP with respect to the related literature. 
In~\cite{cesare} an adaptive sampling algorithm is presented which can be used for estimating the best sampling frequency for energy hungry sensors. However, their approach assumes that the sensors have uniform energy profiles.

Work presented in \cite{kansal_harvest} proposes a harvest-aware adaptive sampling approach to dynamically identify the maximum duty cycle. However, their focus is not on signal approximation from the network.


In~\cite{reboul}, a Bayesian estimation technique is presented to estimate the wind speed and wind direction signals. The authors have supplemented their estimation using the assumption that the wind speed and wind direction signals have a correlation with hourly tide data. However, in our work we assume that signals are compressible due to the presence of spatial-temporal correlation among the data collected at different sensing nodes.

A number of studies~\cite{Spatiotemporal,backstack,gupta} have utilized the spatial-temporal correlation of the signal to reduce sampling requirements. Though our approach has similar assumption, we have considered non-uniform energy profile of the sensors, which is different. Moreover, we have used Sparse Random Projections, which is also different from these approaches.

A large number of signal approximation techniques use Compressive Sensing~\cite{wakin,rana2011adaptive} to conserve transmission energy assuming that radio is the dominant component of energy consumption, however we assume energy-hungry sensor dominates the energy consumption.
A slightly different compressive sensing based data gathering approach is presented in~\cite{interplay} which investigates the impact of a routing topology generated sparse projection matrix on the accuracy of the approximation. Our work is different from theirs since our projection matrix is not based on the routing topology rather it is populated based on the energy profile of the sensors.
A more general sparse projection matrix proposed in the compressive sensing literature can be found in~\cite{Baraniuk:2008p10398}.
However, the non-zero elements of this sparse projection matrix are chosen uniform randomly, but, in order to enable energy-aware sampling, the non-zero elements in the projection matrix need to be chosen nonuniform randomly.

Recently, we have extended the theory of compressive sensing showing that it can be used to support non-uniform sampling~\cite{6146546,shen2013nonuniform}. However, in this paper we choose sparse random projections over compressive sensing, since the decoding process of compressive sensing is computationally expensive. The complexity of decoding a $n$ data point vector is $O(n^3)$. Whereas, decoding complexity of sparse random projections is as low as $O(mn\log{n})$, where $m$ is the number of projections. In this paper the projects are generated locally without any coordination between basestation and nodes, and, the final signal recovery takes place at the resource enriched basestation. In our future study we seek to conduct the signal recovery at the resource limited sensor nodes. For that purpose a low complexity decoder will be very useful.

\section{Conclusion}
\label{sec:conclusion}
In this paper we propose \EP, which offers an optimal sensing strategy for any rechargeable wireless sensor networks with energy hungry sensors. We derive the upper bound of reconstruction error for \EP. Then using two energy hungry sensor (wind speed and wind direction sensor) data from Springbrook National park sensor deployment, we evaluate the validity of the upper bound of the reconstruction error. We compare \EP with our previously proposed signal reconstruction method EAST. The shortcoming of EAST is that it does not guarantee the optimal utilization of node energy, which could potentially compromise reconstruction accuracy. The comparison results presented in this paper clearly show that \EP can significantly reduce the reconstruction error. We also present a distributed algorithm for \EP, which can be readily used in the sensor network deployments. In our future study we want to study the adaptation of dynamic change in energy profile. 
%
%

\begin{spacing}{0.8}
\bibliographystyle{IEEEtran}
\bibliography{reference_2,sigProc}
\end{spacing}

\begin{center}
\newpage
{\bf APPENDIX}
\end{center}
\begin{proof}[Proof of Proposition~\ref{prop:1}] It can be proved that the projection matrix defined by Equation~(\ref{eqn:sparse_approx}) satisfies these conditions:
\begin{eqnarray}
\mathbb{E}\left[\Phi_{ij}\right]=0 \mbox{,}\hspace{.05cm} \mathbb{E}\left[\Phi_{ij}^2\right]=1 \mbox{,}\hspace{.05cm} \mathbb{E}\left[\Phi_{ij}^4\right]=\frac{1}{g_j}.\label{eqn:cond_expection}
\end{eqnarray}
Define independent random variables $w_1,..w_{\ell}$ where, $w_i = \left( \sum_{j=1}^{N} u_j \Phi_{ij} \right) \left( \sum_{j=1}^{N} v_j \Phi_{ij} \right) $
Expectation and second moment of $w_i$ can be computed as,
\begin{eqnarray}
\mathbb{E}[w_i] &=& \mathbb{E}[\sum_{j=1}^{N}u_j v_j\Phi^2_{ij}+\sum_{\ell \neq m} u_{\ell} v_m\Phi_{il}\Phi_{im}]\nonumber \\
&=&\sum_{j=1}^{N}u_j v_j \mathbb{E}[\Phi_{ij}^2]+\sum_{l\neq m}u_{\ell} v_m\mathbb{E}[\Phi_{il}]\mathbb{E}[\Phi_{im}]\nonumber \\
&=& u^Tv. \nonumber
\end{eqnarray}
\begin{eqnarray}
\mathbb{E}[w_i^2] &=& \mathbb{E}[(\sum_{j=1}^{N}u_j v_j\Phi^2_{ij})^2+(\sum_{\ell \neq m} u_{\ell} v_m\Phi_{il}\Phi_{im})^2  \nonumber \\&+&2(\sum_{j=1}^{N}u_j v_j\Phi^2_{ij})(\sum_{\ell \neq m} u_{\ell} v_m\Phi_{il}\Phi_{im})]\nonumber \\
&=&\sum_{j=1}^{N}u_j^2 v_j^2 \mathbb{E}[\Phi_{ij}^4]\nonumber\\&+&2\sum_{\ell < m}u_{\ell}v_{\ell}u_mv_m \mathbb{E}[\Phi_{il}^2]\mathbb{E}[\Phi_{im}^2]\nonumber\\&+&\sum_{\ell \neq m}u_{\ell}^2v_m^2\mathbb{E}[\Phi_{il}^2]\mathbb{E}[\Phi_{im}^2] \nonumber \\&+&2\sum_{\ell < m}u_{\ell}v_{\ell}u_mv_m \mathbb{E}[\Phi_{il}^2]\mathbb{E}[\Phi_{im}^2]\nonumber \\
&=&\sum_{j=1}^{N}\frac{1}{g_j} u_j^2 v_j^2 +2\sum_{\ell \neq m}u_{\ell}v_{\ell}u_mv_m + \sum_{\ell \neq m}u_{\ell}^2 v_m^2 \nonumber\\
&=&2(\sum_{j=1}^{N} u_j^2 v_j^2 +\sum_{\ell \neq m} u_{\ell} v_{\ell} u_m v_m) \nonumber\\&+&(\sum_{j=1}^{N}u_j^2 v_j^2+\sum_{l\neq m} u_{\ell}^2 v_m^2) \nonumber \\&&+\sum_{j=1}^{N}\frac{1}{g_j}u_j^2 v_j^2-3\sum_{j=1}^{N} u_j^2 v_j^2 \nonumber\\
&=& 2(u^Tv)^2+||u||_{2}^2||v||_{2}^{2}+\sum_{j=1}^{N}\frac{1}{g_j}u_j^2 v_j^2\nonumber\\&-&3\sum_{j=1}^{N} u_j^2 v_j^2. \nonumber \end{eqnarray}
Since $x^Ty = \frac{1}{\ell} \sum_{i = 1}^{\ell} w_i$, using the above result we can show that:
\begin{eqnarray}
Var(x^Ty)&=& \frac{1}{\ell}((u^Tv)^2+||u||_{2}^2||v||_{2}^{2}+\sum_{j=1}^{N}\frac{1}{g_j}u_j^2 v_j^2\nonumber\\&-&3\sum_{j=1}^{N} u_j^2 v_j^2)\nonumber.
\end{eqnarray}
\end{proof}

In order to prove proposition~\ref{prop:3}, we need the following lemma.
\begin{lem}
\label{prop:2}
Consider a data vector $u \in \mathbb{R}^{N}$ which satisfies condition~(\ref{eqn:cond_sparse}). Let $v\in \mathbb{R}^{N \times N}$. Consider a sparse random matrix $\Phi \in \mathbb{R}^{\ell \times N}$ satisfies condition~(\ref{eqn:cond_expection}), with sparsity parameter $\rho=g_j$. Define $\ell= 48\frac{(2+\mu^2\max_j\frac{1}{g_j})k^2{(1+\gamma)\log{(N)}}}{c^2\epsilon^2}$. The random projections $\frac{1}{\sqrt{\ell}}\Phi u$ and $\frac{1}{\sqrt{\ell}}\Phi v_i$ then produces an estimation $\hat{a_i}$ for $u^T v_i$, with probability at least $1-N^{-\gamma}$, satisfying $|\hat{a_i}-u^Tv_i|\leq \epsilon||u||_2||v_i||_2$, $\forall_{1\leq i\leq N}$
\end{lem}
\begin{proof}[Proof of Lemma~\ref{prop:2}]
Consider two vectors $u,v\in \mathbb{R}^{N}$ that satisfies condition (\ref{eqn:cond_sparse}). Set $\ell=\ell_1\ell_2$ where $\ell_1$ and $\ell_2$ are two positive integers. Partition the $\ell\times N$ matrix $\Phi$ into $\ell_2$ matrices {$\Phi_1,...,\Phi_{\ell_2}$} each of size $\ell_1 \times N$. Create the random projections $\{x_1=\frac{1}{\sqrt{\ell_1}}\Phi_1 u,...,x_{\ell_2}=\frac{1}{\sqrt{\ell_1}}\Phi_{\ell_2} u\}$ and $\{y_1=\frac{1}{\sqrt{\ell_1}}\Phi_1 v,...,y_{\ell_2}=\frac{1}{\sqrt{\ell_1}}\Phi_{\ell_2} v\}$

Let us define the independent random variables $z_1,...,z_{\ell_2}$, where $z_{\ell}=x_{\ell}^Ty_{\ell}$. We now apply Proposition~\ref{prop:1} to each $z_{\ell}$ and find that,
$\mathbb{E}[z_{\ell}]=u^Ty$ and
\begin{eqnarray}
\mathbb{E}[z_{\ell}]&=&u^Tv \mbox{ and}\nonumber\\
Var(z_{\ell})&=&\frac{1}{\ell_1}((u^Tv)^2+||u||_{2}^2||v||_{2}^{2}+\sum_{j=1}^{N}\frac{1}{g_j}u_j^2 v_j^2\nonumber\\&-&3\sum_{j=1}^{N} u_j^2 v_j^2).\nonumber
\end{eqnarray}

Using Chebyshev inequality it can be shown that,
\begin{eqnarray}
&P&(|z_{\ell}-u^Tv| \geq \epsilon||u||_2||v||_2) \leq \frac{Var(z_{\ell})}{\epsilon^2||u||_2^2||v||_2^2}\nonumber \\
&=& \frac{1}{\epsilon^2 \ell_1} 
(\frac{(u^Tv)^2}{||u||_2^2||v||_2^2}+\frac{||u||_2^2||v||_2^2}{||u||_2^2||v||_2^2})\nonumber \\&+&\frac{\sum_{j=1}^{N}\frac{1}{g_j}u_j^2 v_j^2-3\sum_{j=1}^{N} u_j^2 v_j^2}{||u||_2^2||v||_2^2}\nonumber \\
& \leq & \frac{1}{\epsilon^2 \ell_1} (1+1+\frac{\mu^2||u||_2^2||v||_2^2\max_j\frac{ 1}{g_j}}{||u||_2^2||v||_2^2})\nonumber \\ && [\text{since} ||u||_{\infty} \leq \mu ||u||_2] \nonumber \\
&=&\frac{1}{\epsilon^2 \ell_1}(2+\mu^2\max_j\frac{ 1}{g_j}) \triangleq \delta. \nonumber
\end{eqnarray}
Therefore, we can obtain a constant probability $\delta$ by setting $\ell_1=O(\frac{2+\mu^2\sum_{j=1}^N\frac{1}{g_j}}{\epsilon^2})$.

Therefore, we can obtain a constant probability, say $\delta=\frac{1}{4}$, by setting $\ell_1=4*\frac{2+\mu^2\max_j\frac{1}{g_j}}{\epsilon^2}$.

Now it can be shown that~\cite{distributed} for any pair of vectors $u$ and $v_i\in\{v_1,..,v_n\}$, the random projections $\frac{1}{\ell} \Phi u$ and $\frac{1}{\ell}\Phi v_i$ produce an estimate $\hat{a_i}$ for $u^Tv_i$ that lies outside the tolerable approximation interval with probability at most  $e^{-c^2\ell_2/12}$, where $0<c<1$ is some constant. Taking union bound over all such vectors, the probability that at least one $\hat{a_i}$ lies outside the tolerable interval with probability is upper bounded by 
\begin{eqnarray}
P_e&\leq& N e^{-c^2\ell_2/12}\nonumber\\
\log{(P_e)} &\leq& \log{N} -c^2\ell_2/12 \hspace{1 cm}\nonumber\\
\mbox{Let}, \log{(P_e)}&=&-\gamma \log{N}\nonumber\\
\log{N}{(1+\gamma)}&\leq& c^2\ell_2/12\nonumber \\                                                                            \ell_2&\geq&\frac{12{(1+\gamma)\log{N}}}{c^2}.\nonumber
\end{eqnarray}
Now, 
\begin{eqnarray}
\ell&=&\ell_1\ell_2\nonumber\\
&=& 4*\frac{2+\mu^2\max_j\frac{1}{g_j}}{\epsilon^2} * \frac{12(1+\gamma)\log{N}}{c^2}\nonumber\\
&=&48\frac{(2+\mu^2\max_j\frac{1}{g_j}){(1+\gamma)\log{N}}}{c^2\epsilon^2}.\nonumber
\end{eqnarray}
\end{proof}

\begin{proof}[Proof of Proposition~\ref{prop:3}]
Consider an orthonormal transform $\Psi\in \mathbb{R}^{N \times N}$. Let us represent the transform coefficients using $\theta=[u^T\psi_1,...,u^T\psi_N]^T$. Reordering the transform coefficients $\theta$ in decreasing of magnitude, i.e., $|\theta|_{(1)}\geq |\theta|_{(2)}....\geq |\theta|_{(N)}$, the approximation error by taking the largest $k$ coefficients in magnitude, and setting the remaining coefficients to zero is given by   $||\theta-\hat{\theta}_{opt}||_2^2 =\sum_{i=k+1}^N|\theta|_{(i)}^2$. Let $||\theta-\theta_{opt}||_2^2 \leq \eta||\theta||_2^2$ and assume that $u$ satisfies condition~(\ref{eqn:cond_sparse}), with positive integer, $\ell= 48\frac{(2+\mu^2\max_j\frac{1}{g_j})k^2{(1+\gamma)\log{N}}}{c^2\beta^2}.$
The random projections $\frac{1}{\sqrt{\ell}}\Phi u$ and $\{\frac{1}{\sqrt{\ell}}\Phi\psi_1,....,.\frac{1}{\sqrt{\ell}}\Phi\psi_n\}$ can produce estimates $\{\hat{\theta_1},...,\hat{\theta_N}\}$, where the estimates satisfy $|\hat{\theta_i}-\theta_i| \leq \beta ||\theta||_2$ with high probability (Lemma~\ref{prop:2}).

It can be shown that~\cite{distributed} for $\beta=O(\frac{\epsilon \eta}{k})$, the approximate error is: $||u-\hat{u}||_2^2=(1+\epsilon)\eta||u||_2^2.$ Therefore the number of random projections can be given by
\begin{eqnarray}
\ell= 48\frac{(2+\mu^2\max_j\frac{1}{g_j})k^2{(1+\gamma)\log{N}}}{c^2\epsilon^2\eta^2}.\nonumber
\end{eqnarray}
\end{proof}

\end{document}